\documentclass[a4paper,11pt]{article}
\usepackage{amsfonts,amsthm}
\usepackage{helvet}
\usepackage{courier}
\usepackage{graphicx}
\usepackage{amssymb}

\usepackage{amsmath}
\usepackage{enumerate}

\usepackage{tikz}
\usepackage{xcolor}
\usepackage{complexity}

\usetikzlibrary{shapes, arrows, shadows, calc, positioning, fit}
\usetikzlibrary{decorations.pathreplacing, decorations.pathmorphing, shapes.misc}
\usetikzlibrary{tikzmark}

\usepackage[ruled,vlined]{algorithm2e}

\newtheorem{theorem}{Theorem}

\newtheorem{lemma}[theorem]{Lemma}
\newtheorem{corollary}[theorem]{Corollary}
\newtheorem{example}[theorem]{Example}

\newcommand{\pref}{{\succ}} 
\newcommand{\capacity}{q}
\newcommand{\quota}{\capacity}
\newcommand{\applicantset}{A}
\newcommand{\courseset}{C}
\newcommand{\applicant}{a}
\newcommand{\course}{c}
\newcommand{\inputsetting}{I}

\newcommand{\weaklypref}{\succeq}
\newcommand{\strictlypref}{\succ}
\newcommand{\tie}{\simeq}

\newcommand{\matching}{\mu}
\newcommand{\copiesorder}{\Sigma}

\newcommand{\eqclass}{C}

\newcommand{\mechanism}{\phi}
\newcommand{\numapplicants}{n_1}
\newcommand{\numcourses}{n_2}

\newcommand{\preflist}{P}
\newcommand{\preflistset}{\mathcal{P}}

\newcommand{\path}{\mathfrak{P}}
\newcommand{\cpath}{\mathfrak{C}}

\newcommand{\visited}{\Upsilon}

\newcounter{nodemarkers}

\newcommand*\circled[1]{\tikz[baseline=(char.base)]{
    \node[shape=circle,draw,inner sep=2pt] (char) {#1};}}

\newcommand{\hide}[1]{}

\def\mt{t\kern-0.035cm\char39\kern-0.03cm}
\def\ml{l\kern-0.035cm\char39\kern-0.03cm}
\def\md{d\kern-0.035cm\char39\kern-0.03cm}

\newtheorem{df}{Definition}

\begin{document}

\title{\bf Pareto Optimal Matchings in Many-to-Many Markets with Ties \thanks{
%
%
This research has been co-financed by the European Union (European Social Fund - ESF) and Greek national funds under Thales grant MIS
380232 (Eirinakis, Magos, Mourtos),
by grant EP/K01000X/1 from the Engineering and Physical Sciences Research Council (Manlove, Rastegari),
grants VEGA 1/0344/14, 1/0142/15 from the Slovak Scientific grant agency VEGA (Cechl\'arov\'a), student grant VVGS-PF-2014-463 (Oce\v l\'akov\'a) and OTKA grant K108383 (Fleiner). 
The authors gratefully acknowledge the support of COST Action IC1205 Computational Social Choice. A preliminary version of this paper has been published in the proceedings of the 8th International Symposium on Algorithmic Game Theory (SAGT 2015) \cite{CEFMMMOR-sagt15}. 
}
}
\author{Katar\'{i}na Cechl\'{a}rov\'{a}$^1$, Pavlos Eirinakis$^2$, Tam\'{a}s Fleiner$^3$, Dimitrios Magos$^4$, \\ David Manlove$^5$, Ioannis Mourtos$^2$, Eva Oce\v l\'akov\'a$^1$ and Baharak Rastegari$^5$ \\
\\
\small
\emph{$^1$ Institute of Mathematics, Faculty of Science, P.J. \u{S}af\'{a}rik University} \\ 
\small
\emph{Ko\u{s}ice, Slovakia}
\\
\small
\emph{$^2$ Department of Management Science and Technology,} \\ 
\small
\emph{Athens University of Economics and Business, Athens, Greece}
\\
\small
\emph{$^3$ Budapest University of Technology and Economics, Magyar tud\'{o}sok k\"{o}r\'{u}tja and} \\ 
\small
\emph{MTA-ELTE Egerv\'ary Research Group, Budapest, Hungary}
\\
\small
\emph{$^4$ Department of Informatics, Technological Educational Institute of Athens} \\ 
\small{Egaleo, Greece}
\\
\small
\emph{$^5$ School of Computing Science, University of Glasgow, Glasgow, UK}
}
\date{}
\maketitle              

\begin{abstract}
We consider Pareto optimal matchings (POMs) in a many-to-many market of applicants and courses where applicants have preferences, which may include ties, over individual courses and lexicographic preferences over sets of courses. Since this is the most general setting examined so far in the literature, our work unifies and generalizes several known results. Specifically, we characterize POMs and introduce the \emph{Generalized Serial Dictatorship Mechanism with Ties (GSDT)} that effectively handles ties via properties of network flows. We show that GSDT can generate all POMs using different priority orderings over the applicants, but it satisfies truthfulness only for certain such orderings. This shortcoming is not specific to our mechanism; we show that any mechanism generating all POMs in our setting is prone to strategic manipulation. 
This is in contrast to the one-to-one case (with or without ties), for which truthful mechanisms generating all POMs do exist. 
\end{abstract}

\noindent
{\bf Keywords:} Pareto optimality, many-to-many matching, serial dictatorship, truthfulness

\section{Introduction}
We study a many-to-many matching market that involves two finite disjoint sets, a set of applicants and a set of courses. Each applicant finds a subset of courses acceptable and has a preference ordering, not necessarily strict, over these courses. Courses do not have preferences. Moreover, each applicant has a quota on the number of courses she can attend, while each course has a quota on the number of applicants it can admit.

A matching is a set of applicant-course pairs such that each applicant is paired only with acceptable courses and the quotas associated with the applicants and the courses are respected. The problem of finding an ``optimal'' matching given the above market is called the \emph{Course Allocation problem} (CA). Although various optimality criteria exist, \emph{Pareto optimality} (or \emph{Pareto efficiency}) remains the most popular one (see, e.g., \cite{AS98,ACMM04,CG10a,CEFMMP14,SS13WINE}). Pareto optimality is a fundamental concept that economists regard as a minimal requirement for a ``reasonable'' outcome of a mechanism. A matching is a Pareto optimal matching (POM) if there is no other matching in which no applicant is worse off and at least one applicant is better off. Our work examines Pareto optimal many-to-many matchings in the setting where applicants' preferences may include ties.

In the special case where each applicant and course has quota equal to one, our setting reduces to the extensively studied \emph{House Allocation problem} (HA) \cite{HZ79,AS98,Zho90},
also known as the \emph{Assignment problem} \cite{Gar73,Bogomolnaia-Moulin}. 
Computational aspects of HA have been examined thoroughly \cite{ACMM04,Man13} and particularly for the case where applicants' preferences are strict. In \cite{ACMM04} the authors provide a characterization of POMs in the case of strict preferences and utilize it in order to construct polynomial-time algorithms for checking whether a given matching is a POM and for finding a POM of maximum size. They also show that any POM in an instance of HA with strict preferences can be obtained through the well-known \emph{Serial Dictatorship Mechanism (SDM)} \cite{AS98}.
SDM is a straightforward greedy algorithm that allocates houses sequentially according to some exogenous priority ordering of the applicants.

Recently, the above results have been extended in two different directions. The first one \cite{KMRZ14} considers HA in settings where preferences may include ties. Prior to \cite{KMRZ14}, few works in the literature had considered extensions of SDM to such settings. The difficulty regarding ties, observed already in \cite{Svensson}, is that the assignments made in the individual steps of the SDM are not unique, and an unsuitable choice may result in an assignment that violates Pareto optimality. In \cite{BM-2004} and \cite{Svensson} an implicit extension of SDM is provided (in the former case for \emph{dichotomous preferences}, where an applicant's preference list comprises a single tie containing all acceptable houses), but without an  explicit  description of  an algorithmic procedure. 
Krysta et al. \cite{KMRZ14} describe a mechanism called the \emph{Serial Dictatorship Mechanism with Ties (SDMT)} that combines SDM with the use of augmenting paths to ensure Pareto optimality. SDMT includes an augmentation step, in which applicants already assigned a house may exchange it for another, equally preferred one, to enable another applicant to take a house that is most preferred given the assignments made so far. They also show that any POM in an instance of HA with ties can be obtained by an execution of SDMT and also provide the so-called \emph{Random Serial Dictatorship Mechanism with Ties (RSDMT)} whose (expected) approximation ratio is $\frac{e}{e-1}$ with respect to the maximum-size POM. 

The second direction \cite{CEFMMP14} extends the results of \cite{ACMM04} to the many-to-many setting (i.e., CA) with strict preferences, while also allowing for a structure of applicant-wise acceptable sets that is more general than the one implied by quotas. Namely, \cite{CEFMMP14} assumes that each applicant selects from a family of course subsets that is downward closed with respect to inclusion. This work provides a characterization of POMs assuming that the preferences of applicants over sets of courses are obtained from their (strict) preferences over individual courses in a lexicographic manner. Using this characterization, it is shown that deciding whether a given matching is a POM can be accomplished in polynomial time. In addition, \cite{CEFMMP14} generalizes SDM to provide the \emph{Generalized Serial Dictatorship Mechanism (GSD)}, which can be used to obtain any POM for CA under strict preferences. The main idea of GSD is to allow each applicant to choose not her most preferred set of courses at once but, instead, only one course at a time (i.e., the most preferred among non-full courses that can be added to the courses already chosen).  This result is important as the version of SDM where an applicant chooses immediately her most preferred set of courses cannot obtain all POMs.

\vspace{-3pt}
\paragraph{\textbf{Our contribution.}}

In the current work, we combine the directions appearing in \cite{KMRZ14} and \cite{CEFMMP14} to 
explore the many-to-many setting in which applicants have preferences, which may include ties, over individual courses. We extend these preferences to sets of courses lexicographically, 
since lexicographic set preferences naturally describe human behavior \cite{GG96}, they have already been considered in models of exchange of indivisible goods \cite{CEFMMP14,Lesca} and also possess theoretically interesting properties including responsiveness \cite{KM01}.
  
We provide a characterization of POMs in this setting, leading to a polynomial-time algorithm for testing whether a given matching is Pareto optimal. We introduce the \emph{Generalized Serial Dictatorship Mechanism with Ties (GSDT)} that generalizes both SDMT and GSD. SDM assumes a priority ordering over the applicants, according to which applicants are served one by one by the mechanism. Since in our setting applicants can be assigned more than one course, each applicant can return to the ordering several times (up to her quota), each time choosing just one course. The idea of using augmenting paths \cite{KMRZ14} has to be employed carefully to ensure that during course shuffling no applicant replaces a previously assigned course for a less preferred one. To achieve this, we utilize methods and properties of network flows. Although we prove that GSDT can generate all POMs using different priority orderings over applicants, we also observe that some of the priority orderings guarantee truthfulness whereas some others do not. That is, there may exist priority orderings for which some applicant  benefits from misrepresenting her preferences. This is in contrast to SDM and SDMT in the one-to-one case in the sense that all executions of these mechanisms induce truthfulness. This shortcoming however is not specific to our mechanism, since we establish that any mechanism generating all POMs is prone to strategic manipulation by one or more applicants. 

\vspace{-3pt}

\paragraph{\textbf{Organization of the paper.}}
In Section~\ref{setting} we define our notation and terminology. The characterization is provided in Section~\ref{characterization}, while GSDT is presented in Section~\ref{GSDT}. A discussion on applicants' incentives in GSDT is provided in Section~\ref{truthfulness}.
Avenues for future research are discussed in Section~\ref{futWor}.
%
\setlength{\textfloatsep}{0pt}
\begin{table}[ht!]
\centering
\begin{tabular}{|c c l| c    c|  }
\hline
applicant &  quota & preference list &  course   & quota    \\
\hline
$\applicant_1$ & 2&   $(\{c_1,c_2\},\{c_3\},\emptyset)$      & $\course_1$    & 2  \\[1mm]
$\applicant_2$ & 3 &  $(\{c_2\}, \{c_1,c_3\}, \emptyset)$       & $\course_2$    & 1 \\[1mm]
$\applicant_3$ & 2 &  $(\{c_3\},\{c_2\},\{c_1\})$       & $\course_3$    & 1   \\  [1mm]
\hline
\end{tabular}
\caption{An instance $I$ of {\sc ca}.}\label{tab1}
\end{table}

\section{Preliminary definitions of notation and terminology}\label{setting}
Let $\applicantset=\{\applicant_1,\applicant_2,\cdots,\applicant_{\numapplicants}\}$ be the set of applicants, $\courseset=\{\course_1,\course_2, \cdots, \course_{\numcourses}\}$ the set of courses and $[i]$ denote the set $\{1,2,\ldots,i\}$. Each applicant $\applicant$ has a quota $b(\applicant)$ that denotes the maximum number of courses $\applicant$ can accommodate into her schedule, and likewise each course $\course$ has a quota $\quota(\course)$ that denotes the maximum number of applicants it can admit.
Each applicant finds a subset of courses acceptable and has a transitive and complete preference ordering, not necessarily strict, over these courses.
We write $\course\strictlypref_{\applicant} \course'$ to denote that applicant $\applicant$ \emph{(strictly) prefers} course $\course$ to course $\course'$, and  $\course \tie_{\applicant} \course'$ to denote that $\applicant$ is \emph{indifferent between} $\course$ and $\course'$. We write $\course \weaklypref_{\applicant} \course'$ to denote that $\applicant$ either prefers $\course$ to $\course'$ or is indifferent between them, and say that $\applicant$ \emph{weakly prefers} $\course$ to $\course'$.

Because of indifference, each applicant divides her acceptable courses into \emph{indifference classes} such that she is indifferent between the courses in the same class and has a strict preference over courses in different classes. Let $\eqclass^{\applicant}_t$ denote the $t$'th indifference class, or \emph{tie}, of applicant $\applicant$ where $t \in [\numcourses]$. 
The preference list of any applicant $\applicant$ is the tuple of sets $\eqclass^{\applicant}_t$, i.e., $\preflist(\applicant) = (\eqclass^{\applicant}_1, \eqclass^{\applicant}_2 , \cdots ,\eqclass^{\applicant}_{\numcourses})$; 
we assume that $\eqclass^{\applicant}_t = \emptyset$ implies $\eqclass^{\applicant}_{t'} = \emptyset$ for all $ t' >t$.
Occasionally we consider $\preflist(\applicant)$ to be a set itself and write $c \in \preflist(\applicant)$ instead of $c \in \eqclass^{\applicant}_t$ for some $t$. We denote by $\preflistset$ the joint preference profile of all applicants, and by $\preflistset(-\applicant)$ the joint profile of all applicants except $\applicant$. 
Under these definitions, an instance of CA is denoted by $\inputsetting=(\applicantset, \courseset, \preflistset, b, q)$. Such an instance appears in Table \ref{tab1}. 


A \emph{(many-to-many) assignment} $\matching$ is a subset of $\applicantset\times \courseset$. For $\applicant \in \applicantset$, $\matching(\applicant)=\{\course\in\courseset: (\applicant,\course)\in \matching\}$ and for $\course \in \courseset$, $\matching(\course)=\{\applicant\in\applicantset: (\applicant,\course)\in\matching\}$. An assignment $\matching$ is a \emph{matching} if $\matching(\applicant) \subseteq \preflist(\applicant)$---and thus $\matching$ is individually rational, $|\matching(\applicant)|\leq b(\applicant)$ for each $\applicant\in\applicantset$ and $|\matching(\course)|\leq q(\course)$ for each $\course\in\courseset$. We say that $\applicant$ is \emph{exposed} if $|\matching(\applicant)| < b(\applicant)$, and is \emph{full} otherwise. 
Analogous definitions of exposed and full hold for courses.

For an applicant $\applicant$ and a set of courses $S$, we define the \emph{generalized characteristic vector} $\chi_{\applicant}(S)$ as the vector $(|S\cap \eqclass_1^{\applicant}|, |S\cap \eqclass_2^{\applicant}|,\ldots,|S\cap \eqclass_{\numcourses}^{\applicant}| )$. We assume that for any two sets of courses $S$ and $U$, $\applicant$ prefers $S$ to $U$ if and only if $\chi_\applicant(S) >_{lex} \chi_\applicant(U)$, i.e., if and only if there is an indifference class $\eqclass_t^{\applicant}$ such that $|S\cap \eqclass_t^{\applicant}|>|U\cap \eqclass_t^{\applicant}|$ and $|S\cap \eqclass_{t'}^{\applicant}|= |U\cap \eqclass_{t'}^{\applicant}|$ for all $t' < t$. 
If $\applicant$ neither prefers $S$ to $U$ nor $U$ to $S$, then she is indifferent between $S$ and $U$. 
We write $S \pref_\applicant U$ if $\applicant$ prefers $S$ to $U$, $S \tie_\applicant U$ if $\applicant$ is indifferent between $S$ and $U$, and $S \weaklypref_{\applicant} U$ if $\applicant$ weakly prefers $S$ to $U$. 


A matching $\matching$ is a \emph{Pareto optimal matching} (POM) if there is no other matching in which some applicant is better off and none is worse off. Formally, $\matching$ is Pareto optimal if there is no matching $\matching'$ such that $\matching'(\applicant) \weaklypref_{\applicant} \matching(\applicant)$ for all $\applicant \in \applicantset$, and $\matching'(\applicant') \pref_{\applicant'} \matching(\applicant')$ for some $\applicant' \in \applicantset$. If such a $\matching'$ exists, we say that $\matching'$ \emph{Pareto dominates} $\matching$.

A \emph{deterministic mechanism} $\mechanism$ maps an instance to a matching, i.e. $\mechanism : \inputsetting  \mapsto \matching$ where $\inputsetting$ is a CA instance and $\matching$ is a matching in $\inputsetting$. 
A \emph{randomized mechanism} $\mechanism$ maps an instance to a distribution over possible matchings.
%
Applicants' preferences are private knowledge and an applicant may prefer not to reveal her preferences truthfully. A deterministic mechanism is 
\emph{dominant strategy truthful} (or just \emph{truthful}) 
if all applicants always find it best to declare their true preferences, no matter what other applicants declare. Formally speaking, for every applicant $\applicant$ and every possible declared preference list $\preflist'(\applicant)$, $\mechanism(\preflist(\applicant),\preflistset(-\applicant))\weaklypref_{\applicant}\mechanism(\preflist'(\applicant),\preflistset(-\applicant))$, for all $\preflist(\applicant), \preflistset(-\applicant)$.
A randomized mechanism $\mechanism$ is \emph{universally truthful}
if it is a probability distribution over deterministic truthful mechanisms.
\section{Characterizing Pareto optimal matchings}\label{characterization}


Manlove \cite[Sec.\ 6.2.2.1]{Man13} provided a characterization of Pareto optimal matchings in HA with preferences that may include indifference. He defined three different types of \emph{coalitions} with respect to a given matching such that the existence of either means that a subset of applicants can trade among themselves (possibly using some exposed course) ensuring that, at the end, no one is worse off and at least one applicant is better off. 
He also showed that if no such coalition exists, then the matching is guaranteed to be Pareto optimal. We show that this characterization extends to the many-to-many setting, although the proof is more complex and involved than in the one-to-one setting. We then utilize the characterization in designing a polynomial-time algorithm for testing whether a given matching is Pareto optimal.

In what follows we assume that in each sequence $\cpath$ no applicant or course appears more than once.

An \emph{alternating path coalition} w.r.t. $\matching$ comprises a sequence 
$\cpath=\langle\course_{j_0}, \applicant_{i_0}, \course_{j_1}, \linebreak \applicant_{i_1},\ldots,\course_{j_{r-1}},\applicant_{i_{r-1}},\course_{j_r}\rangle$ where $r\geq 1$, $\course_{j_k} \in \matching(\applicant_{i_k})$ ($0\leq k\leq r-1$), $\course_{j_k} \not\in\matching(\applicant_{i_{k-1}})$ ($1\leq k \leq r$), $\applicant_{i_0}$ is full, and $\course_{j_r}$ is an exposed course. Furthermore, $\applicant_{i_0}$ prefers $\course_{j_1}$ to $\course_{j_0}$ and, if $r \geq 2$, $\applicant_{i_k}$ weakly prefers $\course_{j_{k+1}}$ to $\course_{j_k}$ ($1\leq k \leq r-1$).

An \emph{augmenting path coalition} w.r.t. $\matching$ comprises a sequence
$\cpath=\langle\applicant_{i_0}, \course_{j_1},\applicant_{i_1}, \linebreak \ldots,\course_{j_{r-1}},\applicant_{i_{r-1}},\course_{j_r}\rangle$ where $r\geq 1$, $\course_{j_k} \in\matching(\applicant_{i_k})$ ($1\leq k\leq r-1$), $\course_{j_k} \not\in\matching(\applicant_{i_{k-1}})$ ($1\leq k \leq r$), $\applicant_{i_0}$ is an exposed applicant, and $\course_{j_r}$ is an exposed course. Furthermore, $\applicant_{i_0}$ finds $\course_{j_1}$ acceptable and, if $r\geq 2$, $\applicant_{i_k}$ weakly prefers $\course_{j_{k+1}}$ to $\course_{j_k}$ ($1\leq k \leq r-1$).


A \emph{cyclic coalition} w.r.t. $\matching$ comprises a sequence
$\cpath=\langle\course_{j_0}, \applicant_{i_0}, \course_{j_1}, \applicant_{i_1},\ldots,\course_{j_{r-1}},\linebreak \applicant_{i_{r-1}}\rangle$ where $r\geq 2$, $\course_{j_k} \in\matching(\applicant_{i_k})$ ($0\leq k\leq r-1$), and $\course_{j_k} \not\in\matching(\applicant_{i_{k-1}})$ ($1\leq k \leq r$). Furthermore, $\applicant_{i_0}$ prefers $\course_{j_1}$ to $\course_{j_0}$ and $\applicant_{i_k}$ weakly prefers $\course_{j_{k+1}}$ to $\course_{j_k}$ ($1\leq k \leq r-1$). (All subscripts are taken modulo $r$ when reasoning about cyclic coalitions).

We define an \emph{improving coalition} to be an alternating path coalition, an augmenting path coalition or a cyclic coalition. 

Given an improving coalition $\cpath$, the matching
\begin{equation} 
\matching^{\cpath} = (\matching \setminus \{(\applicant_{i_k}, \course_{j_k}): \delta\leq k \leq r-1\}) \cup \{(\applicant_{i_k}, \course_{j_{k+1}}): 0\leq k \leq r-1\} \}
\label{eq:matching'}
\end{equation}
is defined to be the matching obtained from $\matching$ by \emph{satisfying} $\cpath$ ($\delta =1$ in the case that $\cpath$ is an augmenting path coalition, otherwise $\delta=0$). 

We will soon show that improving coalitions are at the heart of characterizing Pareto optimal matchings (Theorem~\ref{thm:characterization}). The next lemma will come in handy in the proof of Theorem~\ref{thm:characterization} and Theorem~\ref{any-pom}. 
We say that a sequence of applicants and courses is a \emph{pseudocoalition} if it satisfies all conditions of an improving coalition except that some courses or applicants may appear more that once.


\begin{lemma}\label{lem:repetition}
Let a matching $\matching$ in an instance $\inputsetting$ of CA admit a pseudocoalition $K$ of length $\ell$ for some finite $\ell$. Then $\matching$ admits an improving coalition $\cpath$ of length at most $\ell$.
\end{lemma}
\begin{proof}
We prove this by induction. Obviously $\ell \geq 2$. For the base case, it is easy to see that $K$ itself is an augmenting path coalition (where $r=1$) if $\ell =2$ , an alternating path coalition (where $r=1$) if $\ell =3$, a cyclic coalition (where $r=2$) if $\ell = 4$ and $K$ ends with an applicant, and an augmenting path coalition (where $r=2$) if $\ell = 4$ and $K$ ends with a course. Assume that the claim holds for all pseudocoalitions of length $d$, $d<\ell$. We show that it also holds for any given pseudocoalition $K$ of length $\ell$. In the rest of the proof we show that either $K$ is in fact an improving coalition, or we can derive a shorter pseudocoalition $K'$ from $K$, hence completing the proof.

\vspace{5pt}
\textbf{The case for a repeated course}: Assume that a course $\course$ appears more than once in $K$. We consider two different scenarios.
\begin{enumerate}
	\item If $\course$ is not the very first element of $K$, then sequence $K$ is in the following form where where $\course=\course_{j_x}=\course_{j_y}$.
$$\dots, a_{i_{x-1}}, \overleftrightarrow{c_{j_x}, a_{i_x}, \dots}, c_{j_y}, \dots$$
In this case, we simply delete the portion of $K$ under the arrow. Note that in the new sequence $c_{j_y}$ appears right after after $a_{i_{x-1}}$, but then $c_{j_y}=\course_{j_x}$. Hence we have obtained a shorter sequence, and it is easy to verify that this new sequence is a pseudocoalition.
	\item If $\course$ is the first element of $K$, then $K$ is in the following form where $\course=\course_{j_0}=\course_{j_y}$.
$$\overleftrightarrow{c_{j_0}, a_{i_{0}}, \dots, a_{i_{y-1}}}, c_{j_y}, \dots$$
In this case, we simply only keep the portion of $K$ under the arrow, i.e. our new sequence starts with $c_{j_0}$ and ends with $a_{i_{y-1}}$. The new sequence is shorter and it is easy to verify that it is a pseudocoalition.
\end{enumerate}
 
\textbf{The case for a repeated applicant}: Assume that an applicant $\applicant$ appears more than once in $K$. We consider the three different possible scenarios.
\begin{enumerate}
	\item If $\applicant$ is not the very first element of $K$, nor the last one, then $K$ is in the following form where $\applicant=\applicant_{i_x}=\applicant_{i_y}$.
$$\dots, \course_{j_x}, \overleftrightarrow{a_{i_{x}}, c_{j_{x+1}}, \ldots, a_{i_{y-1}}}, c_{i_{y}}, a_{i_y}, c_{j_{y+1}}, \dots$$
Note that $c_{j_{x+1}} \weaklypref_{a_{i_x}} c_{j_x}$ and $c_{j_{y+1}} \weaklypref_{a_{i_x}} c_{j_y}$ (since $\applicant_{i_x}=\applicant_{j_y}$). We consider two different cases.
		\begin{enumerate}
			\item If $c_{j_{x+1}} \pref_{a_{i_x}} c_{j_y}$, then we simply only keep the portion of $K$ under the arrow and add $c_{j_y}$ to the beginning. That is, the new sequence is $K'=\langle c_{j_y}, a_{i_x}, c_{j_{x+1}}, \linebreak\ldots, a_{i_{y-1}} \rangle$. $K'$ is shorter than $K$ and it is easy to verify that that it is a pseudocoalition.
			\item If $c_{j_y} \weaklypref_{a_{i_x}} c_{j_{x+1}}$, then $c_{j_{y+1}} \weaklypref_{a_{i_x}} c_{j_{x}}$. Then we simply remove the portion of $K$ from $c_{j_{x+1}}$ up to $a_{i_{y}}$.  That is, the new sequence is $K' =\langle \ldots c_{j_x}, a_{i_x}, c_{j_{y+1}}\ldots\rangle$ which is shorter than $K$. Note that if either $c_{j_{x+1}} \pref_{a_{i_x}} c_{j_x}$ or $c_{j_{y+1}} \pref_{a_{i_x}} c_{j_y}$, then $c_{j_{y+1}} \pref_{a_{i_x}} c_{j_{x}}$. Hence it is easy to verify that $K'$ is a pseudocoalition.
		\end{enumerate}
	\item If $\applicant$ is the first element of $K$ but not the last one, then $K$ is in the following form where $\applicant=\applicant_{i_0}=\applicant_{i_y}$. 
	$$a_{i_0}, \overleftrightarrow{c_{j_1}, \ldots, c_{j_{y}}, a_{i_y}}, c_{j_{y+1}}, \dots$$
	Notice that $\applicant_0$ is exposed and finds $c_{j_{y+1}}$ acceptable and is not matched to it under $\matching$. Hence we simply remove the portion of $K$ under the arrow and get $K' =\langle a_{i_0}, c_{j_{y+1}}, \ldots\rangle$. $K'$ is shorter than $K$ and it is easy to verify that it is a pseudocoalition.
	\item If $\applicant$ is the last element of $K$ but not the first one, then $K$ is of the following form where $\applicant=\applicant_{i_x}=\applicant_{i_y}$.
	$$\course_{j_0}, \ldots, \course_{j_x}, \overleftrightarrow{a_{i_{x}}, c_{j_{x+1}}, \ldots, a_{i_{y-1}}}, c_{j_{y}}, a_{i_y}$$
	Note that, as $\applicant_{i_x}=\applicant_{i_y}$, $c_{j_{x+1}} \weaklypref_{a_{i_x}} c_{j_x}$ and $c_{j_{0}} \weaklypref_{a_{i_x}} c_{j_y}$.
Furthermore, 
	$a_{i_x}$ is not matched to $c_{j_0}$ in $\matching$.
	We consider two different cases.
		\begin{enumerate}
			\item If $c_{j_{x+1}} \pref_{a_{i_x}} c_{j_y}$, then we do as we did in Case 1(a). That is, we only keep the portion of $K$ under the arrow and add $c_{j_y}$ to the beginning; the new sequence is $K'=\langle c_{j_y}, a_{i_x}, c_{j_{x+1}}, \ldots, a_{i_{y-1}} \rangle$. $K'$ is shorter than $K$ and it is easy to verify that it is a pseudocoalition.
			\item If $c_{j_y} \weaklypref_{a_{i_x}} c_{j_{x+1}}$, then $c_{j_{0}} \weaklypref_{a_{i_x}} c_{j_{x+1}} \weaklypref_{a_{i_x}} c_{j_x}$. Then we simply remove the portion of $K$ from $c_{j_{x+1}}$ until the end of the sequence. That is, the new sequence is $K' =\langle c_{j_0}, \ldots ,c_{j_x}, a_{i_x}\rangle$. $K'$ is shorter than $K$ and it is easy to verify that it is a pseudocoalition.
		\end{enumerate}
\end{enumerate}
Note that we do not need to check the case where $K$ starts and ends with an applicant, since then $K$ would not fit the definition of an improving coalition.
\end{proof}

The following theorem gives a necessary and sufficient condition for a matching to be Pareto optimal. 
\begin{theorem}\label{thm:characterization}
Given a CA instance $\inputsetting$, a matching $\matching$ is a Pareto optimal matching in $\inputsetting$ if and only if $\matching$ admits no improving coalition.
\end{theorem}
\begin{proof}
Let $\matching$ be a Pareto optimal matching in $\inputsetting$. Assume to the contrary that $\matching$ admits an improving coalition $\cpath$. It is fairly easy to see
that matching $\matching^{\cpath}$ obtained from $\matching$ according to equation~(\ref{eq:matching'}) Pareto dominates $\matching$, a contradiction.

Conversely, let $\matching$ be a matching in $\inputsetting$ that admits no improving coalition, and suppose to the contrary that it is not Pareto optimal. 
Therefore, there exists some matching $\matching' \neq \matching$ such that $\matching'$ Pareto dominates $\matching$. Let $G_{\matching,\matching'} = \matching \oplus \matching'$ be the graph representing the symmetric difference of $\matching$ and $\matching'$. $G_{\matching,\matching'}$ is hence a bipartite graph with applicants in one side and courses in the other; by abusing notation, we may use $\applicant$ or
$\course$ to refer to a node in $G_{\matching,\matching'}$ corresponding to applicant $\applicant\in\applicantset$ or course $\course\in\courseset$, respectively.
Note that each edge of $G_{\matching,\matching'}$ either belongs to $\matching$, referred to as a $\matching$-edge, or to $\matching'$, referred to as a $\matching'$-edge.

We first provide an intuitive sketch of the rest of the proof.  Our goal is to find an improving coalition, hence contradicting the assumption that $\matching$ does not admit one. We start with an applicant $\applicant_{i_0}$ who strictly prefers $\matching'$ to $\matching$. We choose a $\matching'$-edge $(\applicant_{i_0},\course_{j_1})$ such that $\course_{j_1}$ belongs to the first indifference class where $\matching'$ is better than $\matching$ for $\applicant_{i_0}$. Either this edge already represents an improving coalition---which is the case if $\course_{j_1}$ is exposed in $\matching$---or there exists a $\matching$-edge incident to $\course_{j_1}$. In this fashion, we continue by adding edges that alternatively belong to $\matching'$ (when we are at an applicant node) and $\matching$ (when we are at a course node), never passing the same edge twice (by simple book-keeping). 
Most importantly, we always make sure that when we are at an applicant node, we choose a $\matching'$-edge such that the applicant weakly prefers the course incident to the $\matching'$-edge to the course incident to the precedent $\matching$-edge. We argue that ultimately we either reach an exposed course (which implies that an augmenting or an alternating path coalition exists in $\matching$) or are able to identify a cyclic coalition in $\matching$.  

Let us first provide some definitions and facts.
Let $\applicantset'$ be the set of applicants who prefer $\matching'$ to $\matching$; i.e., $\applicantset' = \{\applicant | \matching'(\applicant) \pref_{\applicant} \matching(\applicant)\}$. 
Moreover, for each $\applicant\in\applicantset$, let $\matching'_{\setminus\matching}(\applicant) = \matching'(\applicant)\setminus\matching(\applicant)$ and $\matching_{\setminus\matching'}(\applicant) = \matching(\applicant)\setminus\matching'(\applicant)$. Likewise, for each course $\course\in\courseset$, let $\matching_{\setminus\matching'}(\course) = \matching(\course)\setminus\matching'(\course)$ and $\matching'_{\setminus\matching}(\course) = \matching'(\course)\setminus\matching(\course)$.
Note that these sets will be altered during the course of the proof.
In what follows and in order to simplify presentation, we will say 
that we remove a $\matching'$-edge $(\applicant,\course)$ from $G_{\matching,\matching'}$ to signify that we remove $\course$
from $\matching'_{\setminus\matching}(\applicant)$ and remove 
$\applicant$ from $\matching'_{\setminus\matching}(\course)$; similarly, we say that we remove a $\matching$-edge $(\applicant,\course)$ from $G_{\matching,\matching'}$
to signify that we remove $\course$
from $\matching_{\setminus\matching'}(\applicant)$ and remove 
$\applicant$ from $\matching_{\setminus\matching'}(\course)$.
The facts described below follow from the definitions and the assumption that $\matching'$ Pareto dominates $\matching$. Let $\visited$ be a set containing applicants (initially, $\visited =\emptyset$); we will explain later what this set will come to contain.

\vspace{5pt}
\noindent
\textbf{Fact 1:} For every course $\course$ that is full under $\matching$, it is the case that $|\matching_{\setminus\matching'}(\course)| \geq |\matching'_{\setminus\matching}(\course)|$.


\vspace{5pt}
\noindent
\textbf{Fact 2:} For every applicant $\applicant \in \applicantset' \setminus \visited$, there exists $\ell^*_{\applicant}\leq \numcourses$ such that $|\eqclass^{\applicant}_{\ell} \cap \matching'_{\setminus\matching}(\applicant)| = |\eqclass^{\applicant}_{\ell} \cap \matching_{\setminus\matching'}(\applicant)|$ for all $\ell<\ell^*_{\applicant}$, and $|\eqclass^{\applicant}_{\ell_{\applicant}^*} \cap \matching'_{\setminus\matching}(\applicant)| > |\eqclass^{\applicant}_{\ell_{\applicant}^*} \cap \matching_{\setminus\matching'}(\applicant)|$. Hence, $\eqclass^{\applicant}_{\ell_{\applicant}^*}$ denotes the first indifference class of $\applicant$ in which $\matching'$ is better than $\matching$ for $\applicant$. 

\vspace{5pt}
\noindent
\textbf{Fact 3:} For every applicant $\applicant \in \applicantset'$ that is full in $\matching$, there exists $\ell^+$, $\ell^*_{\applicant} < \ell^+\leq \numcourses$, such that 
$|\eqclass^{\applicant}_{\ell^+} \cap \matching_{\setminus\matching'}(\applicant)| > |\eqclass^{\applicant}_{\ell^+} \cap \matching'_{\setminus\matching}(\applicant)|$. 

\vspace{5pt}

In what follows, we iteratively remove edges from  $G_{\matching,\matching'}$ (so as not to visit an edge twice and maintain the aformentioned facts)
and create a corresponding path $\cpath$,
which we show that in all possible cases implies 
an improving coalition in $\matching$, 
thus contradicting our assumption.

As $\matching'$ Pareto dominates $\matching$, $\applicantset'$ is nonempty. Let $\applicant_{i_0}$ be an applicant in $\applicantset'$. By Fact~2, there exists a course $\course_{j_1} \in \eqclass_{\ell_{\applicant_{i_0}}^*} \cap \matching'_{\setminus\matching}(\applicant_{i_0})$.
Hence, we remove edge $(\applicant_{i_0}, \course_{j_1})$ from $G_{\matching,\matching'}$ and add it to $\cpath$, i.e.,
$\cpath = \langle\applicant_{i_0}, \course_{j_1}\rangle$.
%
%
If $\course_{j_1}$ is exposed, then an augmenting path coalition is implied. More specifically, if $\applicant_{i_0}$ is also exposed in $\matching$, then $\cpath = \langle\applicant_{i_0},\course_{j_1}\rangle$ is an augmenting path coalition. Otherwise, if $\applicant_{i_0}$ is full in $\matching$, then it follows from Fact 3 that there exists a course $\course_{j_0}\in \matching_{\setminus\matching'}(\applicant_{i_0})$ such that $\course_{j_1} \pref_{\applicant_{i_0}} \course_{j_0}$. Then $\cpath=\langle\course_{i_0}, \applicant_{i_0}, \course_{j_1}\rangle$ is an alternating path coalition.

If, on the other hand, $\course_{j_1}$ is full,
we continue our search for an improving coalition in an iterative manner (as follows) until we either reach an exposed course (in which case we show that an augmenting or alternating coalition is found) or revisit an applicant that belongs to $\applicantset'$ (in which case we show that a cyclic coalition is found). Note that $\applicant_{i_0}\in\matching'_{\setminus\matching}(\course_{j_1})$, hence it follows from Fact 1 that $|\matching_{\setminus\matching'}(\course_{j_1})| \geq 1$.
At the start of each iteration $k\geq 1$, we have: 
\[\cpath=\langle\applicant_{i_0}, \course_{j_1}, \ldots, \applicant_{i_{k-1}}, \course_{j_{k}}\rangle\]
such that by the construction of $\cpath$ (through the iterative procedure) every applicant $\applicant_{i_x}$ on $\cpath$ weakly prefers the course that follows her on $\cpath$--- i.e., $\course_{j_{x+1}}$ to which she is matched to in $\matching'$ but not in $\matching$---to the course that precedes her---i.e., $\course_{j_{x}}$ to which she is matched to in $\matching$ but not in $\matching'$. Moreover, $\applicant_{i_0}$ is either exposed and can accommodate one more course, namely $\course_{j_1}$, or is full and hence there exists a course $\course_{j_0}\in\matching_{\setminus\matching'}(\applicant_{i_0})$ such that $\applicant_{i_0}$ strictly prefers $\course_{j_1}$ to $\course_{j_0}$. Notice that $\cpath$ would imply an improving coalition (using Lemma~\ref{lem:repetition}) if $\course_{j_k}$ is exposed.

During the iterative procedure, set $\visited$ includes those applicants who strictly prefer the course they are given in $\cpath$ under $\matching'$ than under $\matching$. That is, $\visited=\{\applicant_{i_k}\in\cpath: \course_{j_{k+1}}\succ_{\applicant_{i_k}} \course_{j_k}\}\cup \{\applicant_{i_0}\}$.
Note that applicant $\applicant_{i_0}$ is also included in $\visited$, since either she is exposed, or $\exists \course_{j_0}\in\matching_{\setminus\matching'}(\applicant_{i_0})$ such that $\course_{j_1} \pref_{\applicant_{i_0}}\course_{j_0}$. 


\vspace{5pt}
\noindent
\textbf{The Iterative Procedure}: We repeat the following procedure until either the last course on $\cpath$, $\course_{j_k}$, is exposed or Case 3 is reached.
If $\course_{j_k}$ is full, then by Fact 1 there must exist an applicant $\applicant_{i_k}\in \matching_{\setminus\matching'}(\course_{i_k})$; note that $\applicant_{i_{k}} \neq \applicant_{i_{k-1}}$. Remove $(\applicant_{i_k},\course_{j_k})$ from $G_{\matching,\matching'}$.
Let $\eqclass^{\applicant_{i_k}}_{\ell}$ be the indifference class of $\applicant_{i_k}$ to which $\course_{j_k}$ belongs. We consider three different cases.
\begin{itemize}
	\item \textbf{Case 1: $|\eqclass^{\applicant_{i_k}}_{\ell} \cap \matching'_{\setminus\matching}(\applicant_{i_k})| \geq |\eqclass^{\applicant{i_k}}_{\ell} \cap \matching_{\setminus\matching'}(\applicant_{i_k})|$.} Then there must exist a course $\course_{j_{k+1}} \in \eqclass^{\applicant_{i_k}}_{\ell} \cap \matching'_{\setminus\matching}(\applicant_{i_k})$; note that $\course_{j_k} \neq \course_{j_{k+1}}$. 
Hence, we remove $(\applicant_{i_k},\course_{j_{k+1}})$ from $G_{\matching,\matching'}$ and append $\langle  \applicant_{i_k}, \course_{j_{k+1}}\rangle$  to $\cpath$.
	%
	 
	\item \textbf{Case 2: $|\eqclass^{\applicant_{i_k}}_{\ell} \cap \matching'_{\setminus\matching}(\applicant_{i_k})| < |\eqclass^{\applicant{i_k}}_{\ell} \cap \matching_{\setminus\matching'}(\applicant_{i_k})|$ and $\applicant_{i_k} \notin \visited$.} Recall that $\matching'$ Pareto dominates $\matching$, hence it must be that $\applicant_{i_k} \in \applicantset'$. Therefore if follows from Fact 2 that $\ell^*_{\applicant_{i_k}} < \ell$ and there exists a course $\course_{j_{k+1}} \in \eqclass^{\applicant_{i_k}}_{\ell^*_{\applicant_{i_k}}} \cap \matching'_{\setminus\matching}(\applicant_{i_k})$; note that $\course_{j_k} \neq \course_{j_{k+1}}$. 
Hence, we remove $(\applicant_{i_k},\course_{j_{k+1}})$ from $G_{\matching,\matching'}$ and append $\langle  \applicant_{i_k}, \course_{j_{k+1}}\rangle$  to $\cpath$.
Moreover, we add $\applicant_{i_k}$ to $\visited$. 

	\item \textbf{Case 3: $|\eqclass^{\applicant_{i_k}}_{\ell} \cap \matching'_{\setminus\matching}(\applicant_{i_k})| < |\eqclass^{\applicant{i_k}}_{\ell} \cap \matching_{\setminus\matching'}(\applicant_{i_k})|$ and $\applicant_{i_k} \in \visited$}; note that this case cannot happen when $k=1$. Let $\applicant_{i_z}$ be the first occurrence of $\applicant_{i_k}$ which has resulted in $\applicant_{i_k}$ being added to $\visited$; i.e., $0 \leq z<k$ and $\applicant_{i_z} = \applicant_{i_k}$. Therefore, $\cpath$ is of the form $\langle\applicant_{i_0}, \course_{j_1}, \applicant_{i_1},\ldots, \applicant_{i_z}, \course_{j_{z+1}}, \ldots \applicant_{i_{k-1}}, \course_{j_k}\rangle$. Let us consider $\cpath=\langle\course_{j_k}, \applicant_{i_z}, \course_{j_{z+1}}, \ldots \applicant_{i_{k-1}}\rangle$. 
Since $\applicant_{i_z} = \applicant_{i_k}\in\visited$,  $\applicant_{i_z} \in \applicantset'$ and, by the construction of $\cpath$, $\course_{j_{z+1}}\in \eqclass^{i_k}_{\ell^*_{i_k}}$, 
it follows from Fact 2 that $\ell^*_{a_{i_z}} < \ell$, which implies
(using the assumption of Case 3) 
that $\course_{j_{z+1}} \pref_{i_z} \course_{j_k}$.
	
\end{itemize}

\vspace{5pt}
\noindent

If at any iteration $k$ we find a $\course_{j_k}$  that is  exposed in $\matching$ or we arrive at Case 3, then 
an improving coalition is implied (using Lemma~\ref{lem:repetition}), which contradicts our assumption on
$\matching$. Otherwise (cases 1 and 2), we continue with a new iteration. 
However, since the number of edges is bounded, this iterative procedure is bound to terminate, either by reaching an exposed course or by reaching Case 3.
\end{proof}

It follows from Theorem \ref{thm:characterization} that in order to test whether a given matching $\matching$ is Pareto optimal, it is sufficient to check whether it does not admit any improving coalition. To check this, we construct the \emph{extended envy graph} associated with $\matching$.

\begin{df}
\label{EEG}The \emph{extended envy graph} $G(\matching)=(V_{G(\matching)},E_{G(\matching)})$
associated with a matching $\matching$ is a weighted digraph with $V_{G(\matching)}=V_\applicantset\cup V_\courseset\cup
V_{\matching}$, where $V_\courseset=\{\course: \course\in \courseset\}, V_\applicantset = \{\applicant: \applicant\in \applicantset\}$ and $V_{\matching} = \{\applicant\course:(\applicant,\course)\in \matching\}$, and $E_{G(\matching)}=E_{G(\matching)}^{1}\cup E_{G(\matching)}^{2}\cup
E_{G(\matching)}^{3}$ where:

\begin{itemize}
\item $E_{G(\matching)}^{1}=\{(\course,\applicant): \applicant\in\applicantset, \course\in\courseset,$ and $|\matching(\course)|<q(\course)\}\bigcup\{(\course,\applicant\course'): \course\in \courseset$ and $|\matching(\course)|<q(\course)\},$ with weight equal to $0$;

\item $E_{G(\matching)}^{2}=\{(\applicant,\course): \applicant\in\applicantset, \course\in P(\applicant)\setminus \matching(\applicant)$ and $|\matching(\applicant)|<b(\applicant)\}\bigcup\{(\applicant,\applicant'\course): \applicant\in \applicantset\setminus \{\applicant'\}, \course\in P(\applicant)\setminus \matching(\applicant)$ and $|\matching(\applicant)|<b(\applicant) \},$ with weight equal to $-1$; and

\item $E_{G(\matching)}^{3}=\{(\applicant\course,\course'): \course'\in P(\applicant)\setminus \matching(\applicant)$ and $\course'\succeq_\applicant \course\}
\bigcup\{(\applicant\course,\applicant'\course'): \course'\in P(\applicant)\setminus \matching(\applicant)$ and $\course'\succeq_\applicant \course\},$ with weight equal to $0$ if $\course' \simeq_\applicant \course$ or $-1$ if $\course' \succ_\applicant \course$.

\end{itemize}
\end{df}

That is, $G(\matching)$ includes one node per course, one node per applicant and one node per applicant-course pair that participates in $\matching$.
With respect to $E_{G(\matching)}$, a `$0$' (`$-1$') arc is an arc with weight $0$ ($-1$).
There is a `$0$' arc leaving from each exposed course node to any other
(non-course) node in $G(\matching)$ (arcs in $E_{G(\matching)}^{1}$).
Moreover, there is a $-1$ arc leaving each exposed applicant node $\applicant$ towards any course node or applicant-course node whose corresponding course is in $P(\applicant)\setminus\matching(\applicant)$ (arcs in $E_{G(\matching)}^{2}$).
Finally, there is an arc leaving each applicant-course node $\applicant\course$ towards any course
node or applicant-course node whose corresponding course $\course'$ is not in $\matching(\applicant)$ 
and $\course'\succeq_{\applicant} \course$. The weight of such an arc is $0$ if $\course' \simeq_\applicant \course$ or $-1$ if $\course' \succ_\applicant \course$ (arcs in $E_{G(\matching)}^{3}$).

%
The following theorem establishes the connection between Pareto optimality of a given matching and the absence of negative cost cycles in the extended envy graph corresponding to that matching.

\begin{theorem}\label{eeg-pom}
A matching $\matching$ is Pareto optimal if and only if its extended envy graph $G(\matching)$ has no negative cost cycles.
\end{theorem}
\begin{proof}
First we show that if $G(\matching)$ has no negative cost cycle, then $\matching$ is Pareto optimal.
Assume for a contradiction that $\matching$ is not Pareto optimal; this means that there exists an improving coalition 
$\mathfrak{\courseset}$. 
We examine each of the three different types of improving coalition separately.
Assume that $\mathfrak{\courseset}$ is an alternating path coalition, i.e.,
$\mathfrak{\courseset} = \langle \course_{j_0}, \applicant_{i_0}, \course_{j_1}, \applicant_{i_1},\ldots,\course_{j_{r-1}}\applicant_{i_{r-1}},\course_{j_r}\rangle$ 
where $r\geq 1$, $(\applicant_{i_k},\course_{j_k}) \in \matching$ ($0\leq k\leq r-1$), $(\applicant_{i_{k-1}},\course_{j_k}) \not\in\matching$ ($1\leq k \leq r$), $\applicant_{i_0}$ is full, and $\course_{j_r}$ is an exposed course.
Furthermore, $\course_{j_1}\succ_{\applicant_{i_0}}\course_{j_0}$ and, if $r\geq 2$, $\course_{j_{k+1}}\succeq_{\applicant_{i_k}} \course_{j_k}$ ($1\leq k \leq r-1$).
By Definition~\ref{EEG}, there exists in $G(\matching)$ an arc 
  $(\applicant_{i_k}\course_{j_k},\applicant_{i_{k+1}}\course_{j_{k+1}})$ for each $0\leq k \leq r-2$,
an arc 
$(\applicant_{i_{r-1}}\course_{j_{r-1}},\course_{j_{r}})$, and an arc $(\course_{j_{r}},\applicant_{i_0}\course_{j_0})$, thus creating a cycle $K$ in $G(\matching)$.
Since all arcs of $K$ have weight $0$ or $-1$ and 
$(\applicant_{i_0}\course_{j_0},\applicant_{i_{1}}\course_{j_{1}})$ has weight $-1$ (because $\course_{j_1}\succ_{\applicant_{i_0}}\course_{j_0}$), $K$ is a negative cost cycle, a contradiction.
A similar argument can be used for the case where $\cpath$ is an augmenting path coalition;
our reasoning on the implied cycle would start with a `$-1$' arc $(\applicant_{i_0},\applicant_{i_1}\course_{j_1})$
and would conclude with a `$0$' arc $(\course_{j_{r}},\applicant_{i_0})$.
 Finally, for a cyclic coalition $\mathfrak{\courseset}$, the implied negative cost cycle includes all arcs $(\applicant_{i_k}\course_{j_k},\applicant_{i_{k+1}}\course_{j_{k+1}})$  for $0\leq k \leq r-1$ taken modulo $r$, where at least one arc has weight $-1$ (otherwise, $\mathfrak{\courseset}$ would not be a cyclic coalition). In all cases, the existence of an improving coalition implies a negative cost cycle in $G(\matching)$, a contradiction.

Next, we show that if $\matching$ is Pareto optimal, $G(\matching)$ has no negative cost cycles. We prove this via contradiction.
Let $\matching$ be a POM, and assume that $G(\matching)$ has a negative cost cycle.
Given that there exists a negative cost cycle in $G(\mu)$, there must also exist one that is simple and of minimal length. 
%
Hence, without loss of generality, let $K$ be a minimal (length) simple negative cost cycle in $G(\matching)$. 
We commence by considering that $K$ contains only nodes in $V_\matching$.
We create a sequence $\mathcal{K}$ corresponding to $K$ as follows.
Starting from an applicant-course node of $K$ with an outgoing '$-1$' arc, say $\applicant_{i_0}\course_{j_0}$, we first place $\course_{j_0}$ in $\mathcal{K}$ followed by $\applicant_{i_0}$, and continue by repeating this process for all nodes in $K$.
It is then easy to verify, by Definition \ref{EEG} and the definition of a cyclic coalition, that $K$ corresponds to a sequence
$\mathcal{K}=\langle\course_{j_0},\applicant_{i_0},\ldots,\course_{j_{r-1}},\applicant_{i_{r-1}}\rangle$
that
satisfies all conditions of a cyclic coalition with the exception that some courses or applicants may appear more than once. But this means that, by Lemma~\ref{lem:repetition}, $\mu$ admits an improving coalition, which contradicts that $\mu$ is a POM.

To conclude the proof, we consider the 
case where $K$ 
includes at least one node from $V_{\applicantset}$ or $V_{\courseset}$.
Note that if $K$ contains a node from $V_\applicantset$,
then it also contains a node from $V_\courseset$. This is because, by Definition~\ref{EEG}, the only incoming arcs to an applicant node are of type $E^1_{G(\matching)}$. 
Therefore $K$ includes at least one course node.
We now 
claim that since $K$ is of minimal length, it will contain only one course node. Assume for a contradiction that $K$ contains more than one course nodes. Since $K$ is a negative cost cycle, it contains at least one `$-1$' arc which leaves an applicant node $\applicant$ (or an applicant-course node $\applicant\course'$). Let $\course$ be the first course node in $K$ that appears  after this `$-1$' arc. Recall that for $\course$ to be in $K$, $\course$ is exposed in $\matching$ and there exists a `$0$' arc from $\course$ to any applicant or applicant-course node of $G(\matching)$. Hence, there also exists such an arc from $\course$ to $\applicant$ (or to $\applicant\course'$), thus forming a negative cost cycle  smaller than $K$, a contradiction. 
Note that this also implies that 
$K$ contains at most one node in $V_\applicantset$. This is because,
by Definition~\ref{EEG}, the only incoming arcs to nodes in $V_\applicantset$ are from nodes in $V_\courseset$,
therefore the claim follows from the facts that $\course$ is the only course node in $K$, and $K$ is a simple cycle. 
We create a sequence $\mathcal{K}$ corresponding to $K$ as follows. If $K$ contains an applicant node $\applicant$, we start $\mathcal{K}$ with $\applicant$. If $K$ contains a course node $\course$, 
we end $\mathcal{K}$ with $\course$. Any applicant-course node $\applicant'\course'$ is handled by placing 
$\course'$ first and then $\applicant'$.
It is now easy to verify, by Definition \ref{EEG} and the definition of a cyclic coalition, that $K$ corresponds to a sequence  
$\mathcal{K}=\langle\applicant_{i_0},\course_{j_1},\applicant_{i_1},\ldots,\course_{j_{r-1}},\applicant_{i_{r-1}},\course_{j_r}\rangle$ (if $K$ contains an applicant node) or a sequence  $\mathcal{K}=\langle\course_{j_0},\applicant_{i_0},\ldots,\course_{j_{r-1}},\applicant_{i_{r-1}},\course_{j_r}\rangle$ (if $K$ contains no applicant node), such that $\mathcal{K}$ satisfies all conditions of either an augmenting path coalition (if $K$ contains an applicant node) or an alternating path coalition (if $K$ contains no applicant node), with the exception that some courses or applicants may appear more than once. But this means that, by Lemma~\ref{lem:repetition}, $\mu$ admits an augmenting path coalition (or an alternating path coalition), which contradicts that $\mu$ is a POM.
\end{proof}

Theorem~\ref{eeg-pom} implies that to test whether a matching $\matching$ is Pareto optimal or not, it suffices to create the corresponding extended envy graph and check whether it contains a negative cost cycle. It is easy to see that we can create $G(\matching)$ in polynomial time. To test whether $G(\matching)$ admits a negative cycle, we can make use of various algorithms that exist in the literature (see \cite{CG-99} for a survey), e.g. Moore-Bellman-Ford which terminates in $O(|V_{G(\matching)}||E_{G(\matching)}|)$ time. However, as all arcs in $G(\matching)$ are of costs either $0$ or $-1$ and hence integer, we can use a faster algorithm of \cite{G-92} which, for our setting, terminates in $O(\sqrt{|V_{G(\matching)}|}|E_{G(\matching)}|)$ time.


\begin{corollary}
Given a CA instance $\inputsetting$ and a matching $\matching$, we can check in polynomial time whether $\matching$ is Pareto optimal in $\inputsetting$.
\end{corollary}

\section{Constructing Pareto optimal matchings}\label{GSDT}
We propose an algorithm for finding a POM in an instance of CA, which is in a certain sense a generalization of Serial Dictatorship thus named \emph{Generalized Serial Dictatorship Mechanism with ties} (GSDT).  The algorithm starts by setting the quotas of all applicants to $0$ and those of courses at the original values given by $q$. At each stage $i$, the algorithm selects a single applicant whose original capacity has not been reached, and increases only her capacity by $1$. The algorithm terminates after $B=\sum_{a\in A} b(a)$ stages, i.e., once the original capacities of all applicants have been reached. In that respect, the algorithm assumes a `multisequence' $\Sigma=(a^1,a^2,\dots,a^B)$ of applicants such that each applicant $a$ appears $b(a)$ times in $\Sigma$; e.g., for the instance of Table \ref{tab1} and the sequence $\Sigma=(\applicant_1,\applicant_1,\applicant_2,\applicant_2,\applicant_3,\applicant_2,\applicant_3)$, the vector of capacities evolves as follows:
$$(0,0,0), (1,0,0), (2,0,0), (2,1,0), (2,2,0), (2,2,1), (2,3,1), (2,3,2).$$

Let us denote the vector of applicants' capacities in stage $i$ by $b^i$, i.e., $b^0$ is the all-zeroes vector and $b^B=b$. Clearly, each stage corresponds to an instance $I^i$ similar to the original instance except for the capacities vector $b^i$. At each stage $i$, our algorithm obtains a matching $\matching^i$ for the instance $I^i$. Since the single matching of stage $0$, i.e., the empty matching, is a POM in $I^0$, the core idea is to modify $\matching^{i-1}$ in such way that if $\matching^{i-1}$ is a POM with respect to $I^{i-1}$ then $\matching^i$ is a POM with respect to $I^i$. To achieve this, the algorithm relies on the following flow network. 

Consider the digraph $D=(V,E)$. Its node set is $V=A\cup T\cup C\cup \{\sigma, \tau\}$ where $\sigma$ and $\tau$ are the source and the sink and vertices in $T$ correspond to the ties in the preference lists of all applicants; i.e., $T$ has a node $(a,t)$ per applicant $\applicant$ and tie $t$ such that $\eqclass^a_t \neq \emptyset$. Its arc set is $E=E_1\cup E_2\cup E_3\cup E_4$ where $E_1=\{(\sigma,a): a\in A\}$, $E_2=\{(a,(a,t)): a\in A, \eqclass^a_t \neq \emptyset \}$, $E_3=\{((a,t),c): c\in \eqclass^a_t\}$ and $E_4=\{(c,\tau): c\in C\}$. The graph $D$ for the instance of Table \ref{tab1} appears in Figure \ref{f1}, where an oval encircles all the vertices of $T$ that correspond to the same applicant, i.e., one vertice per tie.

Using digraph $D=(V,E)$, we obtain a flow network $N^i$ at each stage $i$ of the algorithm, i.e., a network corresponding to instance $I^i$, by appropriately varying the capacities of the arcs. (For an introduction on network flow algorithms see, e.g., \cite{Ahuja-etal}.) The capacity of each arc in $E_3$ is always $1$ (since each course may be received at most once by each applicant) and the capacity of an arc $e=(c,\tau)\in E_4$ is always $q(c)$. The capacities of all arcs in $E_1 \cup E_2$ are initially $0$ and, at stage $i$, the capacities of only certain arcs associated with applicant $\applicant^i$ are increased by $1$. For this reason, for each applicant $\applicant$ we use the variable $curr(\applicant)$ that indicates her `active' tie; initially, $curr(\applicant)$ is set to $1$ for all $\applicant \in \applicantset$.

\setlength{\textfloatsep}{10pt}
\begin{figure}[ht]
\vspace{-.5cm}
\begin{center}
	\includegraphics[width=10cm]{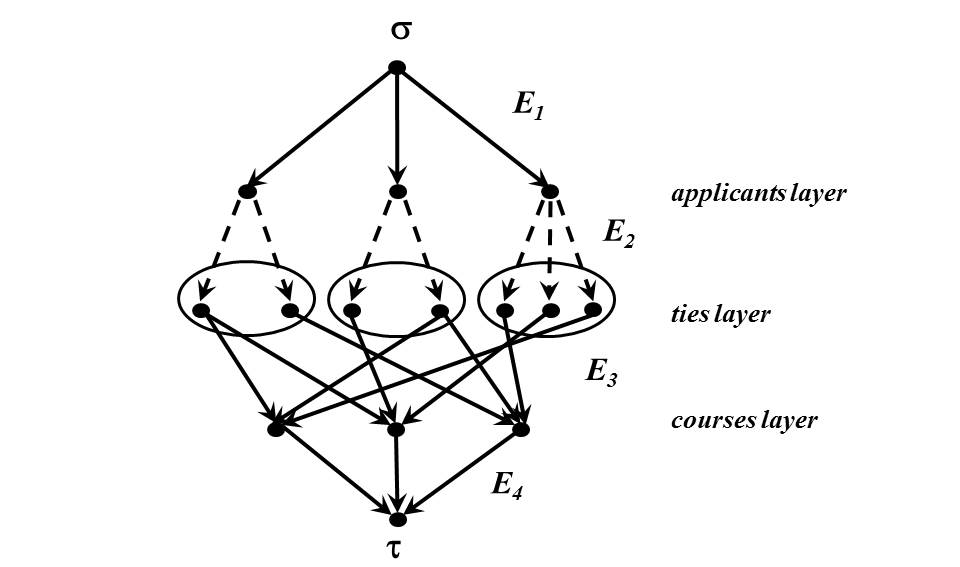}
\end{center}
\caption{Digraph $D$ for the instance $I$ from Table \ref{tab1}.}\label{f1}
\end{figure}

In stage $i$, the algorithm computes a maximum flow $f^i$ whose saturated arcs in $E_3$ indicate the corresponding matching $\matching^i$. The algorithm starts with $f^0=0$ and $\matching^0 = \emptyset$. Let the applicant $\applicant^i \in A$ be a copy of applicant $\applicant$ considered in stage $i$. The algorithm increases by $1$ the capacity of arc $(\sigma, \applicant) \in E_1$ (i.e., the applicant is allowed to receive an additional course). It then examines the tie $curr(\applicant)$ to check whether the additional course can be received from tie $curr(\applicant)$. To do this, 
the capacity of arc $(\applicant, (\applicant,curr(\applicant))) \in E_2$ is increased by $1$. 
The network in stage $i$ where tie $curr(a^i)$ is examined is denoted by $N^{i,curr(a^i)}$. If there is an augmenting $\sigma-\tau$ path in this network, the algorithm augments the current flow $f^{i-1}$ to obtain $f^i$, accordingly augments $\matching^{i-1}$ to obtain $\matching^i$ (i.e., it sets $\matching^i$ to the symmetric difference of $\matching^{i-1}$ and all pairs $(\applicant,\course)$ for which there is an arc $((\applicant,t),\course)$ in the augmenting path) and proceeds to the next stage. Otherwise, it decreases the capacity of $(\applicant, (\applicant,curr(\applicant)))$ by $1$ (but not the capacity of arc $(\sigma, \applicant)$) and it increases $curr(\applicant)$ by $1$ to examine the next tie of $\applicant$; 
if all (non-empty) ties have been examined, the algorithm proceeds to the next stage without augmenting the flow. 
Note that an augmenting $\sigma-\tau$ path in the network $N^{i,curr(a^i)}$ corresponds to an augmenting path coalition in $\matching^{i-1}$ with respect to $\inputsetting^{i}$.

A formal description of GSDT is provided by Algorithm \ref{alg1}, where $w(e)$ denotes the capacity of an arc $e \in E$ and $\oplus$ denotes symmetric difference. Observe that all arcs in $E_2$ are saturated, except for the arc corresponding to the current applicant and tie, thus any augmenting path has one arc from each of $E_1$, $E_2$ and $E_4$ and all other arcs from $E_3$; as a consequence, the number of courses each applicant receives at stage $i$ in any tie cannot decrease at any subsequent step. Also, $\matching^i$ dominates $\matching^{i-1}$ with respect to instance $I^i$ if and only if there is a flow in $N^i$ that saturates all arcs in $E_2$.

\begin{algorithm}[h]
{\bf Input:} an instance $I$ of {\sc CA} and a multisequence $\Sigma$ \\
$f^0:=0$; $\matching^0 := \emptyset$; \\
{\bf for each} $\applicant \in \applicantset$, $curr(\applicant):=1$; \\
{\bf for } $i=1,2,\dots, B$ {\bf do }\\
\{\\
\-\hspace{0.6cm} consider the applicant $\applicant=\applicant^i$; \\
\-\hspace{0.6cm} $w(\sigma, \applicant)$++; \\
\-\hspace{0.6cm} $\path := \emptyset$;

\-\hspace{0.6cm} {\bf while } $\path=\emptyset$ {\bf and} $curr(\applicant)\leq \numcourses$  {\bf and}  $\eqclass^{\applicant}_{curr(\applicant)} \ne \emptyset$ {\bf do }\\
\-\hspace{0.6cm} \{ \\
\-\hspace{1.2cm}   $w(\applicant,(\applicant,curr(\applicant)))$++;\\
\-\hspace{1.2cm} $\path :=$ augmenting path in $N^{i,curr(a)}$ with respect to $f^{i-1}$;\\
\-\hspace{1.2cm} {\bf if } $\path=\emptyset$ {\bf  then} \{ $w(\applicant,(\applicant,curr(\applicant)))$-\,-; $curr(a)$++\};\\
\-\hspace{0.6cm} \} \\

\-\hspace{0.6cm} {\bf if } $\path\ne \emptyset$ {\bf then } \{ $f^i:=f^{i-1}\oplus \path$; \\
\-\hspace{2.9cm} $\matching^i:=\matching^{i-1}\oplus \{(\applicant,\course):((\applicant,t),\course) \in \path $ for some tie $t\}$; \} \\
\-\hspace{0.6cm} {\bf otherwise } \{ $f^i:=f^{i-1}$; $\matching^i:=\matching^{i-1}$; \} \\
\}\\
{\bf return} $\mu^B$;\\
\caption{Producing a POM for any instance of {\sc ca}}\label{alg1}
\end{algorithm}

To prove the correctness of GSDT, we need two intermediate lemmas. Let $e_t \in \mathbb{R}^{n_2}$ be the vector having $1$ at entry $t$ and $0$ elsewhere.
\begin{lemma}\label{lem1}
Let $N^{i,t}$ be the network at stage $i$ while tie $t$ of applicant $\applicant^i$ is examined. Then, there is an augmenting path with respect to $f^{i-1}$ in $N^{i,t}$ if and only if there is a matching $\matching$ such that
$$\chi_\applicant(\matching(\applicant))=\chi_\applicant(\matching^{i-1}(\applicant)) \mbox{\ for each\ } \applicant\ne \applicant^i \mbox{\ and\ } \chi_{\applicant^i}(\matching(\applicant^i))=\chi_{\applicant^i}(\matching^{i-1}(\applicant^i))+  e_t.$$
\end{lemma}
\begin{proof}
Note that the flows $f^{i-1}$ and $f^i$, corresponding to $\matching^{i-1}$ and $\matching^i$ respectively, are feasible in $N^{i}$. Moreover, $f^{i-1}$ is feasible in $N^{i,t}$ and, if there is an augmenting path with respect to $f^{i-1}$ in $N^{i,t}$, then $f^i$ is feasible in $N^{i,t}$ too.

If there is a path augmenting $f^{i-1}$ by $1$ in $N^{i,t}$ thus obtaining $f^{i}$, the number of courses assigned per tie and applicant remain identical except for $\applicant^i$ that receives an extra course from tie $t$. Thus, $\chi_{\applicant}(\matching(\applicant))$ is identical to $\chi_{\applicant}(\matching^{i-1}(\applicant))$ for all $\applicant \in \applicantset$ except for $\applicant^i$ whose characteristic vector has its $t$'th entry  increased by $1$.

Conversely, if the above equation holds, flow $f^i$ is larger than $f^{i-1}$ in the network $N^{i,t}$. Thus standard network flow theory implies the existence of a path augmenting $f^{i-1}$. As a side remark, since the network $N^{i,t}$ has similar arc capacities with the $N^{i-1}$ except for the arcs $(\sigma, \applicant^i)$ and $(\applicant^i,(\applicant^i,t))$, this augmenting path includes these two arcs. 
 
\end{proof}

\begin{lemma}\label{lem2}
Let $S\succeq_{\applicant} U$ and $|S|\ge |U|$. If $c_S$ and $c_U$ denote a least preferred course of applicant $a$ in $S$ and $U$, respectively, then  $S\backslash\{c_S\}\succeq_{\applicant}  U\backslash\{c_U\}$.
\end{lemma}
\begin{proof}
For convenience, we denote $S\backslash\{c_S\}$ by $S'$ and $U\backslash\{c_U\}$ by $U'$. Let $$\chi_{\applicant}(S)=(s_1,s_2,\ldots, s_{n_2}), \ \chi_{\applicant}(U)=(u_1,u_2,\ldots, u_{n_2})$$
and
$$\chi_{\applicant}(S')=(s'_1,s'_2,\ldots, s'_{n_2}),\ \chi_{\applicant}(U')=(u'_1,u'_2,\ldots, u'_{n_2}).$$ 
If $S \simeq_{\applicant} U$, then $c_S\simeq_a c_U$ and $\chi_{\applicant}(S)=\chi_{\applicant}(U)$, therefore $\chi_{\applicant}(S')=\chi_{\applicant}(U')$ and $S' \simeq_{\applicant} U'$. Otherwise, $S \succ_{\applicant} U$ implies that 
there is $k \in [n_2]$ such that $s_j=u_j$ for each $j < k$ and $s_{k}>u_{k}$. 

If there is $j\geq k+1$ with $s_j>0$, then $c_S\in C_r^a$ with $r\geq k+1$. It follows that $s'_k=s_k>u_k\geq u'_k$ while $s'_j \geq u'_j$ for all $j<k$, thus $\chi_a(S')>_{lex}\chi_a(U')$, i.e., $S' \succ_{\applicant} U'$. The same follows if $s_{k} \geq u_{k}+2$ because then $s'_k \geq s_k-1 > u_k\geq u'_k$.

It remains to examine the case where $s_j=0$ for all $j\geq k+1$ and $s_{k}= u_{k}+1$. In this case, $c_S\in C_k^a$  thus $s'_k=s_k-1=u_k$, while $|S|\ge |U|$ implies that either $u_j=0$ for all $j \geq k+1$ or there is a single $k'>k$ such that $u_{k'}=1$. In the former case $c_U\in C_r^a$ for $r\le k$ thus $\chi_a(S')>_{lex}\chi_a(U')$, whereas in the latter one $u'_k=u_k=s_k-1=s'_k$ and $s'_j=u'_j=0$ for all $j>k$ hence $\chi_a(S')= \chi_a(U')$.
 \end{proof}

\begin{theorem}\label{thm:alg1-pom} 
For each $i$, the matching $\matching^{i}$ obtained by GSDT is a POM for instance $I^{i}$.
\end{theorem}
\begin{proof}
We apply induction on $i$. Clearly, $\matching^0=\emptyset$ is the single matching in $I^0$ and hence a POM in $I^0$. We assume that $\matching^{i-1}$ is a POM in $I^{i-1}$ and prove that $\matching^i$ is a POM in $I^i$.

Assume to the contrary that  $\matching^i$ is not a POM in $I^i$. This implies that there is a matching $\xi$ in $I^i$ that dominates $\matching^i$. Then, for all $\applicant \in \applicantset$, $\xi(\applicant) \succeq_{\applicant} \matching^{i}(\applicant) \succeq_{\applicant} \matching^{i-1}(\applicant)$. Recall that the capacities of all applicants in $I^i$ are as in $I^{i-1}$ except for the capacity of $\applicant^i$ that has been increased by $1$. Hence, for all $\applicant \in A \setminus \{\applicant_i\}$, $|\xi(\applicant)|$ does not exceed the capacity of $\applicant$ in instance $I^{i-1}$, namely $b^{i-1}(\applicant)$, while $|\xi(\applicant^i)|$ may exceed $b^{i-1}(\applicant^i)$ by at most $1$.

Moreover, it holds that $|\xi(\applicant^i)| \geq |\matching^{i}(\applicant^i)|$. Assuming to the contrary that $|\xi(\applicant^i)| < |\matching^{i}(\applicant^i)|$ yields that $\xi$ is feasible also in instance $I^{i-1}$. In addition, $|\xi(\applicant^i)| < |\matching^{i}(\applicant^i)|$ implies that it cannot be $\xi(\applicant^i) \simeq_{\applicant^i} \matching^{i}(\applicant^i)$ thus, together with $\xi(\applicant^i) \succeq_{\applicant^i} \matching^{i}(\applicant^i)\succeq_{\applicant^i} \matching^{i-1}(\applicant^i)$, it yields $\xi(\applicant^i) \succ_{\applicant^i} \matching^{i}(\applicant^i) \succeq_{\applicant^i} \matching^{i-1}(\applicant^i)$. But then $\xi$ dominates $\matching^{i-1}$ in $I^{i-1}$, a contradiction to $\matching^{i-1}$ being a POM in $I^{i-1}$.

Let us first examine the case in which GSDT enters the `while' loop and finds an augmenting path, hence $\matching^i$ dominates $\matching^{i-1}$ in $I^i$ only with respect to applicant $\applicant^i$ that receives an additional course. This is one of her worst courses in $\matching^i(\applicant^i)$ denoted as $c_\matching$. Let $c_\xi$ be a worst course for $\applicant^i$ in $\xi(\applicant^i)$. Let also $\xi'$ and $\matching'$ denote $\xi \setminus \{(\applicant^i,c_\xi)\}$ and $\matching^i \setminus \{(\applicant^i,c_\matching)\}$, respectively. Observe that both $\xi'$ and $\matching'$ are feasible in $I^{i-1}$, while having shown that $|\xi(\applicant^i)| \geq |\matching^{i}(\applicant^i)|$ implies through Lemma \ref{lem2} that $\xi'$ weakly dominates $\matching'$ which in turn weakly dominates $\matching^{i-1}$ by Lemma \ref{lem1}. Since $\matching^{i-1}$ is a POM in $I^{i-1}$, $\xi'(\applicant) \simeq_{\applicant} \matching'(\applicant)\simeq_{\applicant} \matching^{i-1}(\applicant)$ for all $\applicant \in \applicantset$, therefore $\xi$ dominates $\matching^i$ only with respect to $\applicant^i$ and $c_\xi \succ_{\applicant^i} c_{\matching}$.  Overall, $\xi(\applicant) \simeq_{\applicant} \matching^i(\applicant)\simeq_{\applicant} \matching^{i-1}(\applicant)$ for all $\applicant \in \applicantset \setminus \{\applicant^i\}$ and $\xi(\applicant^i) \succ_{\applicant^i} \matching^i(\applicant^i)\succ_{\applicant^i} \matching^{i-1}(\applicant^i)$. 

Let $t_\xi$ and $t_\matching$ be the ties of applicant $\applicant^i$ containing $c_\xi$ and $c_\matching$, respectively, where $t_\xi<t_\matching$ because $c_\xi \succ_{\applicant^i} c_{\matching}$. Then, Lemma \ref{lem1} implies that there is a path augmenting $f^{i-1}$ (i.e., the flow corresponding to $\matching^{i-1}$) in the network $N^{i,t_\xi}$. Let also $t'$ be the value of $curr(\applicant^i)$ at the beginning of stage $i$. Since we examine the case where GSDT enters the `while' loop and finds an augmenting path, $\eqclass^{\applicant^i}_{t'} \neq \emptyset$. Thus, $t'$ indexes the least preferred tie from which $\applicant^i$ has a course in $\matching^{i-1}$. The same holds for $\xi'$ since $\xi'(\applicant^i) \simeq_{\applicant^i} \matching^{i-1}(\applicant^i)$. Because $\xi'$ is obtained by removing from $\applicant^i$ its worst course in $\xi(\applicant^i)$, that course must belong to a tie of index no smaller than $t'$, i.e., $t' \leq t_\xi$. This together with $t_\xi<t_\matching$ yield $t' \leq t_\xi <t_\matching$, which implies that GSDT should have obtained $\xi$ instead of $\matching^i$ at stage $i$, a contradiction. 

It remains to examine the cases where, at stage $i$, GSDT does not enter the `while' loop or enters it but finds no augmenting path. For both these cases, $\matching^i=\matching^{i-1}$, thus $\xi$ dominating $\matching^i$ means that $\xi$ is not feasible in $I^{i-1}$ (since it would then also dominate $\matching^{i-1}$). Then, it holds that $|\xi(\applicant^i)|$ exceeds $b^{i-1}(\applicant^i)$ by $1$, thus $|\xi(\applicant^i)| > |\matching^i(\applicant^i)|$ yielding $\xi(\applicant^i) \succ_{\applicant^i} \matching^i(\applicant^i)$. Let $t_\xi$ be defined as above and $t'$ now be the most preferred tie from which $\applicant^i$ has more courses in $\xi$ than in $\matching^i$. Clearly, $t' \leq t_\xi$ since $t_\xi$ indexes the least preferred tie from which $\applicant^i$ has a course in $\xi$. If $t' < t_\xi$, then the matching $\xi'$, defined as above, is feasible in $I^{i-1}$ and dominates $\matching^{i-1}$ because $\xi'(\applicant^i) \succ_{\applicant^i} \matching^{i-1}(\applicant^i)$, a contradiction; the same holds if $t' = t_\xi$ and $\applicant^i$ has in $\xi$ at least two more courses from $t_\xi$ than in $\matching^i$. Otherwise, $t' = t_\xi$ and $\applicant^i$ has in $\xi$  exactly one more course from $t_\xi$ than in $\matching^i$; that, together with $|\xi(\applicant^i)| > |\matching^i(\applicant^i)|$ and the definition of $t_\xi$, implies that 
the index of the least preferred tie  from which $\applicant^i$ has a course in $\matching^{i-1}$ and, therefore, the value of $curr(\applicant^i)$ in the beginning of stage $i$, is at most $t'$. But then GSDT should have obtained $\xi$ instead of $\matching^i$ at stage $i$, a contradiction. 
 
\end{proof}

The following statement is now direct.
\begin{corollary} 
GSDT produces a POM for instance $I$.
\end{corollary}

To derive the complexity bound for GSDT, let us denote by $L$ the length of the preference profile in $I$, i.e., the total number of courses in the preference lists of all applicants. Notice that $|E_3|=L$ and neither the size of any matching in $I$ nor the total number of ties in all preference lists exceeds $L$.

Within one stage, several searches in the network might be needed to find a tie of the active applicant for which the current flow can be augmented. However, one tie is unsuccessfully explored at most once, hence each search either augments the flow thus adding a pair to the current matching or moves to the next tie. So the total number of searches performed by the algorithm is bounded by the size of the obtained matching plus the number of ties in the preference profile, i.e., it is $O(L)$. A search requires a number of steps that remains linear in the number of arcs in the current network (i.e., $N^{i,curr(a^i)}$), but as at most one arc per $E_1, E_2$ and $E_4$ is used, any search needs $O(|E_3|)=O(L)$ steps. This leads to a complexity bound $O(L^2)$ for GSDT.

The next theorem will come in handy when implementing Algorithm~\ref{alg1}, as it implies that for each applicant $\applicant$ only one node in $T$ corresponding to $\applicant$ has to be maintained at a time.
\begin{theorem}
Let $N^{i,t}$ be the network at stage $i$ while tie $t$ of applicant $\applicant^i$ is examined. Then, there is no augmenting path with respect to $f^{i-1}$ in $N^{i,t}$ that has an arc of the form $((\applicant^j,\ell),\course)$ where $j\leq i$ and $\ell<curr(\applicant^j)$.
\end{theorem}
\begin{proof}
Assume otherwise. Let $\path$ be such an augmenting path that is found in round $i$ and used to obtain $\matching^i$. Hence $\path$ corresponds to an augmenting path coalition $\cpath$ of the following form $\langle\applicant^i,\course_{r},\ldots, \course_s, \applicant^j, \course_q, \ldots \rangle$ where $\course_r$ is in the $t$'th indifference class of $\applicant^i$ and both $\course_s$ and $\course_q$ are in the $\ell$'th indifference class of $\applicant^j$, $\ell<curr(\applicant^j)$. Note that as $\ell\geq 1$ thus $curr(\applicant^j) > 1$. It then follows from the description of Algorithm~\ref{alg1} that either $\applicant^j$ is matched to at least one course in $curr(\applicant^j)$ under $\matching^i$ or she is exposed (which would be the case when $curr(\applicant^j) > \numcourses$ or $\eqclass^{\applicant^j}_{curr(\applicant^j)} =\emptyset$).

We first show that $\applicant^i$ and $\applicant^j$ are not the same applicant. Otherwise, $\cpath' = \langle\applicant^j, \course_q, \ldots \rangle$---obtained from $\cpath$ by discarding all courses and applicants that appear before $\applicant^j$---is an augmenting path coalition w.r.t. $\matching^{i-1}$ in $I^i$. Clearly the matching obtained from $\matching^{i-1}$ by satisfying $\cpath'$ Pareto dominates $\matching^i$, as $\ell < t$, contradicting that $\matching^i$ is a POM in $I^{i}$. 

In the remainder of the proof we assume that $\applicant^i \neq \applicant^j$.
%
Let us first consider the case where $\applicant^j$ is matched to a course in $\eqclass^{\applicant^j}_{curr(\applicant^j)}$ under $\matching^i$, and therefore by Lemma \ref{lem1} to a course $\course\in \eqclass^{\applicant^j}_{curr(\applicant^j)}$ under $\matching^{i-1}$. 
Let $\cpath' = \langle\course,\applicant^j,\course_q,\dots\rangle$, i.e., $\cpath'$ is obtained from $\cpath$ by discarding all courses and applicants that appear before $\applicant^j$ and replacing them with $\course$. Since $\ell < curr(\applicant^j)$, we have $\course_q \pref_{\applicant^j} \course$. 
It is then easy to see that $\cpath'$ is an alternating path coalition w.r.t. $\matching^{i-1}$ in $I^{i-1}$, and hence $\matching^{i-1}$ is not a POM in $I^{i-1}$, a contradiction. 
We now consider the case where $\applicant^j$ is exposed in $\matching^i$, and hence exposed in $\matching^{i-1}$ w.r.t. $I^{i-1}$. Let $\cpath' = \langle\applicant^j,\course_q,\dots\rangle$, i.e., $\cpath'$ is obtained from $\cpath$ by discarding all courses and applicants that appear before $\applicant^j$. It is clear that $\cpath'$ is an augmenting path coalition w.r.t. $\matching^{i-1}$ in $I^{i-1}$, contradicting that $\matching'$ is a POM in $I^{i-1}$.
%
\end{proof}



Next we show that GSDT can produce any POM. Our proof makes use of a subgraph of the extended envy graph of Definition \ref{EEG}. 

\begin{theorem}\label{any-pom}
Given a CA instance $\inputsetting$ and a POM $\matching$, there exists a suitable priority ordering over applicants $\copiesorder$ given which GSDT can produce $\matching$.
\end{theorem} 
\begin{proof}
Given an instance $\inputsetting$ and a POM $\matching$, let $G=(V,E)$ be a digraph such that $V=\{\applicant\course: (\applicant,\course) \in \matching\}$ 
and there is an arc from $\applicant\course$ to $\applicant'\course'$ if $\applicant\neq \applicant'$, $\course'\notin\matching(\applicant)$ and $\course' \weaklypref_{\applicant} \course$. An arc $(\applicant\course,\applicant'\course')$ has weight $-1$ if $\applicant$ prefers $\course'$ to $\course$ and has weight $0$ if she is indifferent between the two courses. Note that $G$ is a subgraph of the extended envy graph $G(\matching)$ introduced in Definition \ref{EEG}. We say that $\applicant\course$ envies $\applicant'\course'$ if $(\applicant\course,\applicant'\course')$ has weight $-1$.

Note that if there exists an applicant-course pair $\applicant\course\in\matching$ and a course $\course'$ such that $\course' \pref_{\applicant} \course$ and $\applicant\course'\notin \matching$, then $\course'$ must be full under $\matching$, or else $\matching$ admits an alternating path coalition and is not a POM.

Moreover, all arcs in any given strongly connected component (SCC) of $G$ have weight $0$. To see this, note that by the definition of a SCC, there is a path from any node to every other node in the component. Hence, if there is an arc in a SCC with weight $-1$, then there must be a cycle of negative weight in that SCC. It is then straightforward to see that, by Lemma~\ref{lem:repetition}, $\matching$ admits a cyclic coalition, a contradiction. 

It then follows that if $c' \pref_{a} c$, then $ac'$ and $ac$ cannot belong to the same SCC. Should that occur, there would be an arc (of weight $0$) from $ac'$ to some vertex $a'c^*$ in the same SCC, implying that $c^* \tie_{a} c'$ and thus $c^* \pref_{a} c$. This in turn would yield that there is an arc of weight $-1$ from $ac$ to $a'c^*$ in this SCC, a contradiction.


We create the condensation $G'$ of the graph $G$. It follows from the definition of SCCs that $G'$ is a DAG. Hence $G'$ admits a topological ordering. Let $X'$ be a reversed topological ordering of $G'$ and $X$ be an ordering of the vertices in $G$ that is consistent with $X'$ (the vertices within one SCC may be ordered arbitrarily). Let $\copiesorder$ be an ordering over applicant copies that is consistent with $X$ (we can think of it as obtained from $X$ by removing the courses). We show that GSDT can produce $\matching$ given $\copiesorder$. Note that technically $\copiesorder$ must contain $b(\applicant)$ copies of each applicant $\applicant$. However, as $\matching$ is Pareto optimal, upon obtaining $\matching$ the algorithm will not be able to allocate any more courses to any of the applicants, even if we append more applicant copies to the end of $\copiesorder$.

We continue the proof with an induction. Recall that $X$ is an ordering over the matched pairs in $\matching$. We let $X(\applicant^i) = \course$ where $\applicant^i\course$ is the $i$th element of $X$. Let $\matching^i$ denote the matching that corresponds to the first $i$ elements of $X$; hence $\matching^{|\matching|} = \matching$. We claim that given $\copiesorder$, GSDT is able to produce $\matching^i$ at stage $i$, after augmenting $f^{i-1}$ through an appropriate augmenting path. 

For the base case, note that $\applicant^1X(\applicant^1)$ does not envy any other vertex in $G$ and hence it can only be that $X(\applicant^1) \in \eqclass^1_{\applicant^1}$. It is then easy to see that the path  $\langle \sigma, \applicant^1, (\applicant^1,1), X(\applicant^1), \tau\rangle$ is a valid augmenting path in $N^{1,1}$ and hence GSDT might choose it.

Assume that GSDT produces $\matching^{\ell}$ at the end of each stage $\ell$ for all $\ell<i$. We prove that it can produce $\matching^i$ at the end of stage $i$. Assume, for a contradiction, that this is not the case. Let $r$ denote the indifference class of $\applicant^i$ to which $X(\applicant^i)$ belongs. Note that course $X(\applicant^i)$ is not full in $I^i$, so if the path $\langle \sigma, \applicant^i, (\applicant^i,r), X(\applicant^i), \tau\rangle$ is not chosen by GSDT, it must be the case that GSDT finds an augmenting path in $N^{i,t}$ for some $t<r$. Let $\cpath$ denote the corresponding augmenting path coalition with respect to the matching $\matching^{i-1}$, which would be of the form
\begin{equation*}
\applicant^i,\course_{j_1},\applicant^{i_1},\ldots,\course_{i_{y-1}},\applicant^{i_{y-1}},\course_{j_y}
\end{equation*}
where $\course_{j_1} \in \eqclass^{t}_{\applicant_i}$ and $\course_{j_y}$ is exposed in $\matching^{i-1}$. 
It follows, from the definition of an augmenting path, that there is an edge $(a^{i_k}c_{j_k}, a^{i_{k+1}}c_{j_{k+1}})$ in $G$ for all $k$, $1\leq k\leq y-2$. Furthermore, as $c_{j_1} \pref_{a^i} \matching(a^i)$, there is an edge of weight $-1$ in $G$ from $\applicant^iX(\applicant^i)$ to $\applicant^{i_1}\course_{j_1}$; therefore $\applicant^{i_1}\course_{j_1}$ belongs to a SCC of higher priority than the one to which $\applicant^iX(\applicant^i)$ belongs.
If $\course_{j_y}$ is exposed in $\matching$, then $\cpath'$ that is obtained by adding $X(\applicant^i)$ to the beginning of $\cpath$ is an alternating path coalition in $\matching$, a contradiction to $\matching$ being a POM. Therefore $\course_{j_y}$ must be full in $\matching$. 

As $\course_{j_y}$ is not full in  $\matching^{i-1}$ and full in $\matching$, there must exist an $\applicant^z$, $z>i$, such that $(\applicant^z,\course_{j_y}) \in \matching$. It follows, from the augmenting path coalition $\cpath$, that there is a path in $G$ from $\applicant^{i_1}\course_{j_1}$ to $\applicant^z\course_{j_y}$. If there is also a path from $\applicant^z\course_{j_y}$ to $\applicant^{i_1}\course_{j_1}$, then the two vertices belong to the same SCC; as $\applicant^{i_1}\course_{j_1}$ belongs to a SCC of higher priority than the one to which $\applicant^iX(\applicant^i)$ belongs,
so must $\applicant^z\course_{j_y}$, implying that $\applicant^z\course_{j_y}$ must have appeared before $\applicant^iX(\applicant^i)$ in $X$, a contradiction to $z>i$. If there is no such a path, then $\applicant^z\course_{j_y}$ belongs to a SCC that is prioritized even over the SCC to which $\applicant^{i_1}\course_{j_1}$ belongs, and hence must have appeared before $\applicant^iX(\applicant^i)$ in $X$, a contradiction to $z>i$. 
\end{proof}
%
%
\section{Truthfulness of mechanisms for finding POMs}\label{truthfulness}
It is well-known that the SDM for HA is truthful, regardless of the given priority ordering over applicants. We will show shortly that GSDT is not necessarily truthful, but first prove that this property does hold for some priority orderings over applicants.


\begin{theorem}\label{thm:alg1-truth}
GSDT is truthful given $\copiesorder$ if, for each applicant $\applicant$, all occurrences of $\applicant$ in $\copiesorder$ are consecutive.
\end{theorem}
\begin{proof} 
Without loss of generality, let the applicants appear in $\copiesorder$ in the following order:
$$\underbrace{a_1,a_1,\dots,a_1}_\text{$b(a_1)$-times},\underbrace{a_2,a_2,\dots,a_2}_\text{$b(a_2)$-times}\dots,
\underbrace{a_{i-1},a_{i-1},\dots,a_{i-1}}_\text{$b(a_{i-1})$-times}, \underbrace{a_i,a_i,\dots,a_i}_\text{$b(a_i)$-times},\dots$$
Assume to the contrary that some applicant benefits from misrepresenting her preferences. Let $\applicant_i$ be the first such applicant in $\copiesorder$ who reports $\preflist'(\applicant_i)$ instead of $\preflist(\applicant_i)$ in order to benefit and $\preflistset'=(\preflist'(\applicant_i),\preflistset(-\applicant_i))$. Let $\matching$ denote the matching returned by GSDT using ordering $\copiesorder$ on instance $I =(A,C,\preflistset,b,q)$ (i.e. the instance in which applicant $a_i$ reports truthfully) and $\xi$ the matching returned by GSDT using $\copiesorder$ but on instance $I' =(A,C,\preflistset',b,q)$. Let $s=(\Sigma_{\ell < i} b(a_\ell)) + 1$, i.e., $s$ is the first stage in which our mechanism considers  applicant $\applicant_i$. Let $j$ be the first stage of GSDT such that $\applicant_i$ prefers $\xi^j$ to $\matching^j$, where $s \leq j < s+b(\applicant_i)$. 

Given that applicants $\applicant_1,\dots,\applicant_{i-1}$ report the same in $I$ as in $I'$ and all their occurrences in $\Sigma$ are before stage $j$, Lemma~\ref{lem1} yields  $\matching^j(\applicant_\ell)\simeq_{\applicant_\ell}\xi^j(\applicant_\ell)$ for $\ell=1,2,\ldots,i-1$. Also $\matching^j(\applicant_\ell)=\xi^j(\applicant_\ell)=\emptyset$ for $\ell = i+1, i+2, \ldots, n_1$, since no such applicant has been considered before stage $j$. But then, all applicants apart from $\applicant_i$ are indifferent between $\matching^j$ and $\xi^{j}$, therefore $\applicant_i$ preferring $\xi^{j}$ to $\matching^j$ implies that $\matching^j$ is not a POM in $I^j$, a contradiction to Theorem~\ref{thm:alg1-pom}. 
 
\end{proof}

The next result then follows directly from Theorem~\ref{thm:alg1-truth}.
\begin{corollary}\label{cor:truthful-quota-one}
GSDT is truthful if all applicants have quota equal to one.
\end{corollary}

There are 
priority orderings for which an applicant may benefit from misreporting her preferences, even if preferences are strict. This phenomenon has also been observed in a slightly different context~\cite{BC-12}. Let us also provide an example.

\begin{example}\label{ex:2(2,1)-2(1,1)}
Consider a setting with applicants $\applicant_1$ and $\applicant_2$ and courses $\course_1$ and $\course_2$, for which $b(\applicant_1)=2$, $b(\applicant_2)=1$, $\quota(\course_1)=1$, and $\quota(\course_2)=1$. 
Let $\inputsetting$ be an instance in which 
$\course_2 \pref_{\applicant_1}\course_1$ and $\applicant_2$ finds only $\course_1$ acceptable. This setting admits two POMs, namely $\matching_1=\{(\applicant_1,\course_2), (\applicant_2,\course_1)\}$ and $\matching_2=\{(\applicant_1,\course_1), (\applicant_1,\course_2)\}$. 
%
GSDT returns $\matching_1$ for $\copiesorder=(\applicant_1, \applicant_2, \applicant_1)$. If $\applicant_1$ misreports by stating that she prefers $\course_1$ to $\course_2$, GSDT 
returns $\matching_2$ instead of $\matching_1$. Since 
$\matching_2 \pref_{\applicant_1} \matching_1$,
GSDT is not truthful given $\copiesorder$.
\end{example}

The above observation seems to be a deficiency of GSDT. We conclude by showing that no mechanism capable of producing all POMs is immune to this shortcoming. 

\begin{figure}[t!]
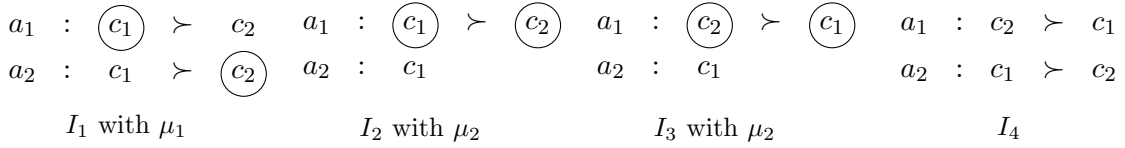
%
\centering
\begin{tabular}{cccc}
	\begin{minipage}{.23\textwidth}
	\setlength{\belowdisplayskip}{0.2cm}
	\[
	\begin{array}{llccc}
	\applicant_1 & : & \textnormal{\circled{$\course_1$}} & \pref & \course_2\\
	\applicant_2 & : & \course_1 & \pref & \textnormal{\circled{$\course_2$}}
	\end{array}\]
	\centering{\small{$\inputsetting_1$ with $\matching_1$}}
	\end{minipage}
&
	\begin{minipage}{.23\textwidth}
	\setlength{\belowdisplayskip}{0.2cm}
	\[
	\begin{array}{llccc}
	\applicant_1 & : & \textnormal{\circled{$\course_1$}} & \pref & \textnormal{\circled{$\course_2$}}\\
	\applicant_2 & : & \course_1 & &\\[0.2cm]
	\end{array}\]
	\centering{\small{$\inputsetting_2$ with $\matching_2$}}
	\end{minipage}
&
	\begin{minipage}{.23\textwidth}
	\setlength{\belowdisplayskip}{0.2cm}
	\[
	\begin{array}{llccc}
	\applicant_1 & : & \textnormal{\circled{$\course_2$}} & \pref & \textnormal{\circled{$\course_1$}}\\
	\applicant_2 & : & \course_1 & &\\[0.2cm]
	\end{array}\]
	\centering{\small{$\inputsetting_3$ with $\matching_2$}}
	\end{minipage}
	&
	\begin{minipage}{.23\textwidth}
	\setlength{\belowdisplayskip}{0.2cm}
	\[
	\begin{array}{llccc}
	\applicant_1 & : & \course_2 & \pref & \course_1\\[0.15cm]
	\applicant_2 & : & \course_1 & \pref & \course_2 \\[0.15cm]
	\end{array}\]
	\centering{\small{$\inputsetting_4$}}
	\end{minipage}
\end{tabular}
\caption{Four instances of CA used in the proof of Theorem~\ref{no-truthful}. In all four instances $b(\applicant_1)=2$, $b(\applicant_2)=1$, $\quota(\course_1)=\quota(\course_2)=1$. For each of instances $\inputsetting_1$ to $\inputsetting_3$, a matching is indicated using circles in applicants' preference lists.}%
\label{fig:2x2}%
\end{figure}

\begin{theorem}\label{no-truthful}
There is no universally truthful randomized mechanism that produces all POMs in CA, even if applicants' preferences are strict and all courses have quota equal to one.
\end{theorem}
\begin{proof}
The instance $\inputsetting_1$ in Figure~\ref{fig:2x2} admits three POMs, namely $\matching_1=\{(\applicant_1,\course_1), \\ (\applicant_2,\course_2)\}$, $\matching_2=\{(\applicant_1,\course_1), (\applicant_1,\course_2)\}$ and $\matching_3=\{(\applicant_1,\course_2), (\applicant_2,\course_1)\}$. Assume a randomized mechanism $\mechanism$ that produces all these matchings. 
Therefore, there must be a deterministic realization of it, denoted as $\mechanism^D$, that returns $\matching_1$ given $\inputsetting_1$. Let us examine the outcome of $\mechanism^D$ under the slightly different applicants' preferences shown in Figure~\ref{fig:2x2}, bearing in mind that $\mechanism^D$ is truthful.
\begin{itemize}
	\item Under $\inputsetting_2$, $\mechanism^D$ must return $\matching_2$. The only other POM under $\inputsetting_2$ is $\matching_3$, but if $\mechanism^D$ returns $\matching_3$, then $\applicant_2$ under $\inputsetting_1$ has an incentive to lie and declare only $\course_1$ acceptable (as in $\inputsetting_2$).
	\item Under $\inputsetting_3$, $\mechanism^D$ must return $\matching_2$. The only other POM under $\inputsetting_3$ is $\matching_3$, but if $\mechanism^D$ returns $\matching_3$, then $\applicant_1$ under $\inputsetting_3$ has an incentive to lie and declare that she prefers $\course_1$ to $\course_2$ (as in $\inputsetting_2$).
 \end{itemize}
$\inputsetting_4$ admits two POMs, namely $\matching_2$ and $\matching_3$. If $\mechanism^D$ returns $\matching_2$, then $\applicant_1$ under $\inputsetting_1$ has an incentive to lie and declare that she prefers $\course_2$ to $\course_1$ (as in $\inputsetting_4$). If $\mechanism^D$ returns $\matching_3$, then $\applicant_2$ under $\inputsetting_3$ has an incentive to lie and declare $\course_2$ acceptable---in addition to $\course_1$---and less preferred than $\course_1$ (as in $\inputsetting_4$).
Thus overall $\mechanism^D$ cannot return a POM under $\inputsetting_4$ while maintaining truthfulness.
 \end{proof}

\section{Future work}\label{futWor}
Our work raises several questions. A particularly important problem is to investigate the expected size of the matching produced by the randomized version of GSDT. It is also interesting to characterize priority orderings that imply truthfulness in GSDT. Consequently, it will be interesting to compute the expected size of the matching produced by 
a randomized GSDT in which the randomization is taken over the priority orderings that induce truthfulness. 

\phantom{\cite{CEFMMMOR-sagt15}}

\bibliographystyle{plain}
\bibliography{bibfile}

%
%

\end{document}